\documentclass[11pt]{article}
\usepackage{amsfonts,amssymb,amsmath,amsthm}
\usepackage{graphicx}
\usepackage{enumerate}

\usepackage{color} % load the color package
   \definecolor{cites}{rgb}{0.50 , 0.00 , 0.00}  % colour for citations
   \definecolor{urls} {rgb}{0.00 , 0.00 , 0.50}  % colour for URL's
   \definecolor{links}{rgb}{0.00 , 0.00 , 0.50}   % colour for links
\usepackage[
      colorlinks=true,   % Yes, we want coloured links!
      citecolor=cites,   % Set the citations colour,
      urlcolor=urls,     % and the URL's colour,
      linkcolor=links,   % and the colour for the links.
      pdftitle={Finite sections of the Fibonacci Hamiltonian},  % Put your title here
      pdfauthor={Marko Lindner, Hagen Söding}, % And your name here
      pdfpagemode=UseOutlines, % None, UseThumbs, UseOutlines, FullScreen
      pdfstartview=FitH,       % Fit, FitH, FitB
      bookmarksopen=false      % Bookmark tree already opened?
   ]{hyperref}

\newcommand\eps\varepsilon
\newcommand\ph\varphi
\newcommand\opsp{{\rm Lim}}

\renewcommand\Re{{\rm Re}}
\renewcommand\Im{{\rm Im}}

\newcommand\vekz[2]{{#1\choose #2}}

\newcommand\R{{\mathbb R}}

\newcommand\Z{{\mathbb Z}}
\newcommand\N{{\mathbb N}}
\newcommand\I{{\mathbb I}}

\newcommand\cI{{\mathcal I}}
\newcommand\cK{{\mathcal K}}
\newcommand\cL{{\mathcal L}}
\newcommand\cF{{\mathcal F}}

\newtheorem{theorem}{Theorem}[section]
\newtheorem{lemma}[theorem]{Lemma}
\newtheorem{corollary}[theorem]{Corollary}

\newtheorem{definition}[theorem]{Definition}

\renewenvironment{proof}
 {\par\noindent{\bf Proof.}}
 {\rule{2mm}{2mm}\pagebreak[2]}

\newenvironment{proofof}[1]
 {\par\noindent{\bf Proof of #1.}}
 {\rule{2mm}{2mm}\pagebreak[2]}

\numberwithin{figure}{section}  % needs package amsmath
\makeatletter
\let\@fnsymbol\@arabic
\makeatother
\begin{document}
\title{\bf Finite sections of the Fibonacci Hamiltonian}
\author{
{\sc Marko Lindner}\footnote{Techn. Univ. Hamburg (TUHH), Institut Mathematik, D-21073 Hamburg, Germany, \url{lindner@tuhh.de}}
\quad and\quad 
{\sc Hagen S\"oding}\footnote{\url{hagen.soeding@tuhh.de}}
}

\date{\today}
\maketitle
\begin{quote}
\renewcommand{\baselinestretch}{1.0}
\footnotesize {\sc Abstract.}
We study finite but growing principal square submatrices $A_n$ of the one- or two-sided infinite Fibonacci Hamiltonian $A$. 
Our results show that such a sequence $(A_n)$, no matter how the points of truncation are chosen, is always stable -- implying that $A_n$ is invertible for sufficiently large $n$ and $A_n^{-1}\to A^{-1}$ pointwise.
\end{quote}

\section{Introduction}
The 1D Schr\"odinger operator $-\Delta+b\cdot$ with a bounded potential $b\in L^\infty(\R)$
can be discretized, via finite differences on a uniform grid on $\R$, by the second order difference operator
\begin{equation} \label{eq:discrSchr}
 (Ax)_n = x_{n-1}+v_nx_n+x_{n+1},\qquad n\in\Z,
\end{equation}
acting on a sequence space like $\ell^p(\Z)$.
The discrete potential $v=(v_n)\in\ell^\infty(\Z)$ corresponds to evaluations of the potential $b$ on the grid (subtracted by a two
that comes from the discretization of the Laplace operator). $A$ is commonly referred to as a discrete 1D Schr\"odinger operator.

A particularly beautiful example, the so-called {\sl Fibonacci Hamiltonian}, arises when the discrete potential $v$
is given by the formula
\begin{equation} \label{eq:mod1}
 v_n = \chi_{[1-\alpha,1)}(n\alpha\!\!\!\mod 1),\qquad n\in\Z,
\end{equation}
where $\alpha=\tfrac{\sqrt 5-1}2$ is the golden ratio and $\chi_\cI$ is the characteristic function of an interval $\cI$.

The sequence $v$ from \eqref{eq:mod1} is not periodic (as $\alpha$ is irrational); it displays a so-called
quasiperiodic pattern. Here are its values $v_1,\dots,v_{55}$ and three attempts to identify basic building blocks 
of the sequence (one with normal/bold face, one with separation by minus signs and one with under/overlines): 
\[
\underline{{\bf 10}110\text-{\bf 10}110\text-110}\text-\underline{{\bf 10}110\text-{\bf 10}110\text-110}\overline{\text-{\bf 10}110\text-110}\text-\underline{{\bf 10}110\text-{\bf 10}110\text-110}\overline{\text-{\bf 10}110\text-101}.
\]
The global pattern of these building blocks (on each scale) is the same as the pattern formed by $1$ and $0$ on the finest scale. The Fibonacci potential shows self-similarity on many levels.

The Fibonacci Hamiltonian is the standard model in 1D for physical properties of so-called quasicrystals
and is therefore heavily studied in mathematical physics.
Most of the research deals with the spectrum of $A$, which is a Cantor set of measure zero without
any eigenvalues (purely singular continuous spectrum).

Our focus is different. The operator \eqref{eq:discrSchr} acts via multiplication
with a two-sided infinite tridiagonal matrix $(a_{ij})_{i,j\in\Z}$. The main diagonal carries the
sequence $v$, the first sub- and super-diagonal are constant to one, and the rest is zero.
We study the applicability of the so-called finite section method to that infinite matrix.

The {\sl finite section method (FSM)} looks at finite submatrices
\[
A_n = (a_{ij})_{i,j=l_n}^{r_n},\qquad n\in\N
\]
of an infinite matrix $A=(a_{ij})_{i,j\in\Z}$ with integer cut-off points
\[
l_n\to -\infty\qquad\text{and}\qquad r_n\to +\infty
\]
and asks whether
\begin{itemize}\itemsep0mm
\item[\bf a)] the matrices $A_n$ are invertible for all sufficiently large $n$ and
\item[\bf b)] their inverses (after embedding them into an infinite matrix again) converge pointwise
in $\ell^p(\Z)$ to the inverse of $A$.
\end{itemize}
Assuming invertibility of $A$ on $\ell^p(\Z)$, property {\bf b)} is equivalent to the uniform boundedness 
of the inverses $A_n^{-1}$. 
%Together with property {\bf a)} this characterizes a stable sequence of operators.

As a consequence, one can solve an infinite system $Ax=b$ approximately by solving
large but finite systems $A_nx_n=b_n$. For one-sided infinite matrices $A=(a_{ij})_{i,j\in\N}$,
all of the above remains true but $l_n$ should be fixed at $1$.

\medskip

{\bf Our result.}
For the Fibonacci Hamiltonian $A$ from \eqref{eq:discrSchr}
with potential \eqref{eq:mod1} as well as for its one-sided infinite submatrix
$A_\N:=(a_{ij})_{i,j\in\N}$, 
we first show that both operators are invertible
on every space $\ell^p(\Z)$, resp. $\ell^p(\N)$, with $p\in[1,\infty]$
before proving that the FSM with arbitrary cut-off points is applicable 
for $A$ as well as for $A_\N$. 

\medskip

{\bf Historic remarks. } %\Marko{Referenzen nachtragen}
Quasicrystals are materials that show features of periodicity (so-called 
Bragg peaks in diffraction experiments) and aperiodicity (symmetries
that rule out periodicity) at the same time -- so-called quasiperiodicity.
They have first been observed by D. Shechtman in 1982 in his laboratory \cite{Steinhardt}
but are meanwhile also found to be ocurring in nature. Physicists and 
mathematicians quickly developed an interest in this topic. In particular 
the spectrum of the corresponding Hamiltonian is of huge interest 
for the understanding of electrical properties of the material \cite{BaakeGrimm,aperiodic,Stollmann}. 
The most famous quasiperiodic ensemble in 2D is the Penrose tiling \cite{Bruijn}.
The understanding of the corresponding spectrum seems however
currently out of reach, so that one resorts to 1D ensembles, the most 
common of which is the Fibonacci sequence \eqref{eq:mod1}.
A detailed history and the current state of the art on the extensively studied
spectral analysis of the Fibonacci Hamiltonian can be found in \cite{DamaEmbrGoro}.

The idea of the FSM is so natural that it
is difficult to give a historical starting point. First rigorous
treatments are from Baxter \cite{Baxter} and Gohberg \& Feldman
\cite{GohbergFeldman} on Wiener-Hopf and convolution operators in
dimension $N=1$ in the early 1960's. For convolution equations in
higher dimensions $N\ge 2$, the FSM goes back to Kozak \& Simonenko
\cite{Kozak,KozakSimonenko}, and for general band-dominated
operators with scalar \cite{RaRoSi1998} and operator-valued
\cite{RaRoSi2001,RaRoSi2001:FSM} coefficients, most results are due
to Rabinovich, Roch \& Silbermann. For the state of the art see e.g.
\cite{BoeGru,Roch:FSM,Li:Habil,Sei:Diss}.

%%%%%%%%%%%%%%%%%%%%%%%%%%%%%%%%%%%%%%%%%%%%%%%%
\section{The finite section method}
%\Marko{Habe noch diesen ersten Satz eingefügt.}
As usual, for an index set $\I\subset\Z$, let $\ell^p(\I)$ denote the set of all complex sequences 
$(x_k)_{k\in\I}$ with $\sum_{k\in\I}|x_k|^p<\infty$ for $p\in[1,\infty)$, and $\ell^\infty(\I)$ be the 
set of all bounded complex sequences over $\I$.

Let $A=(a_{ij})_{i,j\in\Z}$ be a band matrix (i.e.~a matrix with only finitely many nonzero diagonals)
with uniformly bounded complex entries. Then $A$ acts, via matrix-vector multiplication, as a
bounded linear operator on all spaces $\ell^p(\Z)$ with $p\in[1,\infty]$.
Denote that operator again by $A$.

For integer cut-off points $l_1,l_2,\dots$ and $r_1,r_2,\dots$ with
\[
l_n\to -\infty\qquad\text{and}\qquad r_n\to +\infty,
\]
look at the finite submatrices
\begin{equation} \label{eq:An}
A_n = (a_{ij})_{i,j=l_n}^{r_n},\qquad n\in\N
\end{equation}
of $A$ and call the sequence $(A_n)_{n\in\N}$ {\sl stable} if there exists an $n_0\in\N$
such that $A_n$ is invertible for all $n\ge n_0$ and $\sup_{n\ge n_0}\|A_n^{-1}\|<\infty$.

Invertibility of $A$ and stability of $(A_n)$ together are sufficient and necessary for the {\sl applicability} of $(A_n)$,
that is the pointwise convergence
(i.e.~column-wise convergence of the matrices) $A_n^{-1}\to A^{-1}$,
when $A_n^{-1}$ is extended to an infinite matrix again.
%(see e.g. \cite[Theorem 6.1.3]{RaRoSiBook}, \cite[Corollary 1.77]{Li:Book}, \cite[Propositions 1.22, 1.29 and Corollary 1.28]{Sei:Diss}).
This approximation of $A^{-1}$ can be used for solving equations $Ax=b$ approximately
via the solutions of growing finite systems.

We see that it is crucial to know about the stability of $(A_n)$. This stability is closely
connected to a family of one-sided infinite matrices that are associated to $A$ and to the cut-off sequences $(l_{n})$ and $(r_{n})$. Those associated one-sided infinite matrices are partial limits %(in the strong topology) 
of the upper left and the lower right corner of the finite matrix $A_{n}$ as $n\to\infty$. Precisely, the associated matrices are the entrywise limits
\begin{equation} \label{eq:ass_matrix}
(a_{i+l_n',j+l_n'})_{i,j=0}^\infty\ \to\ B_+\qquad\textrm{and}\qquad (a_{i+r_n',j+r_n'})_{i,j=-\infty}^0\ \to\ C_-\qquad\textrm{as}\ n\to\infty
\end{equation}
of one-sided infinite submatrices of $A$, where $(l_n')_{n=1}^\infty$ and $(r_n')_{n=1}^\infty$ are subsequences of $(l_n)_{n=1}^\infty$ and $(r_n)_{n=1}^\infty$, respectively, such that the limits \eqref{eq:ass_matrix} exist. 
The boundedness of the diagonals of $A$ ensures (by Bolzano-Weierstrass and a Cantor diagonal argument) 
the existence of such subsequences and the corresponding limits \eqref{eq:ass_matrix}.
Here is the result.
\begin{lemma}\label{lem:FSMappl}
[Lemma 1.2 of \cite{CWLi:Coburn}]
For a band matrix $A=(a_{ij})_{i,j\in\Z}$ and two cut-off sequences
$(l_n)_{n=1}^\infty$ and $(r_n)_{n=1}^\infty$ in $\Z$ with $l_n\to-\infty$ and $r_n\to+\infty$, the following are equivalent:\\[-7mm]
%\begin{itemize}\itemsep0mm
\begin{enumerate}[\hspace{5mm}(i)]\itemsep-1mm
\item the FSM $(A_n)_{n=1}^\infty$ with $A_n$ from \eqref{eq:An} is applicable to $A$,
\item the FSM $(A_n)_{n=1}^\infty$ with $A_n$ from \eqref{eq:An} is stable,
\item $A$ and the limits $B_+$ and $C_-$ from \eqref{eq:ass_matrix} are invertible for all
subsequences $(l_n')$ of $(l_n)$ and $(r_n')$ of $(r_n)$.
\end{enumerate}
%\end{itemize}
\end{lemma}
So by the choice of the cut-off sequences $(l_{n})$ and $(r_{n})$, one can control the selection of associated matrices $B_{+}$ and $C_-$ 
and hence control the stability and applicability of the FSM. 

The construction of the $B_+$ and $C_-$ brings us to the notion of a 
limit operator \cite{RaRoSi1998,RaRoSiBook,Li:Book}.
\begin{definition}
Let $\I$ be
either $\Z$ or $\N$. For a bounded one- or two-sided infinite band matrix $A=(a_{ij})_{i,j\in\I}$
and a sequence $h_1,h_2,...$ in $\I$ with $|h_n|\to\infty$ we say that $B=(b_{ij})_{i,j\in\Z}$ is
a {\sl limit operator} of $A$ if, for all $i,j\in\Z$,
\begin{equation}\label{eq:limop}
a_{i+h_n,j+h_n}\to b_{ij}\quad\textrm{as}\quad n\to\infty.
\end{equation}
We write $A_h$ instead of $B$. 
\end{definition}

Note that limit operators are always given by a two-sided infinite matrix, no matter if the matrix $A$ to start with is one- or two-sided infinite. 

So in this language, our associated matrices $B_{+}$ and $C_-$ from \eqref{eq:ass_matrix} are one-sided truncations of limit operators of $A$: each $B_+=(b_{ij})_{i,j=0}^\infty$ is a submatrix of a limit operator $B=(b_{ij})_{i,j\in\Z}$
of $A$ w.r.t.~a subsequence $h$ of $(l_n)$, and each $C_-=(c_{ij})_{i,j=-\infty}^0$ is a submatrix of a limit operator $C=(c_{ij})_{i,j\in\Z}$ of $A$ w.r.t.~a subsequence $h$ of $(r_n)$.
To be able to rephrase Lemma \ref{lem:FSMappl} in that language, we introduce
the following notations.

\begin{definition} \label{def:ab}
{\bf a)} For a bounded one- or two-sided infinite band matrix $A=(a_{ij})_{i,j\in\I}$ with $\I\in\{\Z,\N\}$
and a sequence $g_1,g_2,...$ in $\I$ with $|g_n|\to\infty$ we write $\opsp_g(A)$ for the set
of all limit operators $A_h$ with respect to a subsequence $h$ of $g$, and we write $\opsp(A)$
for the set of all limit operators of $A$. 
Moreover, put $\opsp_+(A):=\opsp_{(1,2,3,\dots)}(A)$ and
$\opsp_-(A):=\opsp_{(-1,-2,-3,\dots)}(A)$.

{\bf b)} For a two-sided infinite matrix $A=(a_{ij})_{i,j\in\Z}$, write
%$A_+:=(a_{ij})_{i,j=0}^\infty$ and $A_-:=(a_{ij})_{i,j=-\infty}^0$.
%$A_+:=(a_{ij})_{i,j\in\Z_+}$ and $A_-:=(a_{ij})_{i,j\in\Z_-}$,
$A_\pm:=(a_{ij})_{i,j\in\Z_\pm}$, respectively,
where $\Z_-:=\{\dots,-2,-1,0\}$ and $\Z_+:=\{0,1,2,\dots\}$.
\end{definition}
Note that $A_+$ and $A_-$ overlap in $a_{00}$. Here is the announced reformulation of 
Lemma \ref{lem:FSMappl}.
%\Marko{$A_+$ und $A_-$ sollten so korrekt definiert sein fuer Punkt (iii) unten.}
%
\begin{corollary}\label{cor:FSMappl}
For a bounded band matrix $A=(a_{ij})_{i,j\in\Z}$ and two cut-off sequences $l=(l_n)_{n=1}^\infty$ 
and $r=(r_n)_{n=1}^\infty$ in $\Z$ with $l_n\to-\infty$ and $r_n\to+\infty$, the following are equivalent:\\[-7mm]
%\begin{itemize}\itemsep-1mm
\begin{enumerate}[\hspace{5mm}(i)]\itemsep-1mm
\item[(i)] the FSM $(A_n)_{n=1}^\infty$ with $A_n$ from \eqref{eq:An} is applicable to $A$,
\item[(ii)] the FSM $(A_n)_{n=1}^\infty$ with $A_n$ from \eqref{eq:An} is stable,
\item[(iii)] $A$ and all operators $B_+$ and $C_-$ with $B\in\opsp_l(A)$ and $C\in\opsp_r(A)$ are invertible.
\end{enumerate}
%\end{itemize}
\end{corollary}
If stability holds for $l=(-1,-2,-3,\dots)$ and $r=(1,2,3,\dots)$ then it holds for
arbitrary cut-off sequences $(l_n)$ and $(r_n)$.
\begin{corollary}\label{cor:fullFSM}
For a bounded band matrix $A=(a_{ij})_{i,j\in\Z}$, the following are equivalent:\\[-7mm]
%\begin{itemize}\itemsep-1mm
\begin{enumerate}[\hspace{5mm}(i)]\itemsep-1mm
\item[(i)] the FSM $(A_n)_{n=1}^\infty$ with $A_n$ from \eqref{eq:An} is applicable for arbitrary cut-offs $(l_n)$ and $(r_n)$,
\item[(ii)] the FSM $(A_n)_{n=1}^\infty$ with $A_n$ from \eqref{eq:An} is stable for arbitrary cut-offs $(l_n)$ and $(r_n)$,
\item[(iii)] $A$ and all operators $B_+$ and $C_-$ with $B\in\opsp_-(A)$ and $C\in\opsp_+(A)$ are invertible.
\end{enumerate}
%\end{itemize}
\end{corollary}
The one-sided infinite case, $A=(a_{ij})_{i,j\in\N}$, only requires minimal changes to what
was written above: The sequence $(l_n)$ is then constant at $1$ and therefore the operators $B_+$
do not appear in $(iii)$ of Lemma \ref{lem:FSMappl} and Corollary \ref{cor:FSMappl}.

Limit operators are not only good for detecting
stability\footnote{They only come into play here because the stability of the sequence 
$(A_n)$ is equivalent to the operator $D:=\text{Diag}(A_1,A_2,\dots)$ being a Fredholm 
operator. Then Lemma \ref{lem:Fredholm} below is applied to $D$.} of the FSM.
Their primary purpose is to characterize the coset $A+\cK(X)$ of $A$ modulo the ideal
of all compact operators $\cK(X)$, where we abbreviate $\ell^p(\I)=:X$.

Recall that a bounded linear operator $A$ on $X$, we write $A\in\cL(X)$, is a {\sl Fredholm operator}
if its coset $A+\cK(X)$ is invertible in the so-called Calkin algebra $\cL(X)/\cK(X)$, which holds
iff the nullspace of $A$ has finite dimension and the range of $A$ has finite codimension in $X$.
\begin{lemma} \label{lem:Fredholm}
For a bounded band matrix $A=(a_{ij})_{i,j\in\I}$ with $\I\in\{\Z,\N\}$
and $X=\ell^p(\I)$ with any $p\in[1,\infty]$, it holds that the following are equivalent:\\[-7mm]
%\begin{itemize}\itemsep-1mm
\begin{enumerate}[\hspace{5mm}(i)]\itemsep-1mm
\item $A$ is a Fredholm operator on $X$,
\item  all limit operators of $A$ are invertible on $\ell^p(\Z)$ \cite{RaRoSi1998,LiSei:BigQuest},
\item all limit operators of $A$ are injective on
$\ell^\infty(\Z)$ \cite{CWLi2008:FC,CWLi2008:Memoir}.
\end{enumerate}
%\end{itemize}
\end{lemma}
%

%%%%%%%%%%%%%%%%%%%%%%%%%%%%%%%%%%%%%%%%%%%%%%%%
\section{The Fibonacci word}
%{\tt Erstmal abwarten, was wir weiter unten genau brauchen.}
%
Recall the infinite sequence $v=(v_n)_{n\in\Z}$ of zeros and ones from \eqref{eq:mod1}.
In this section we interpret $v$ as an infinite word over the alphabet $\Sigma=\{0,1\}$.
Let us recall some basic notions on words. For a detailed discussion, including on
the Fibonacci word, see e.g. \cite{LothairAlg}.

\subsection{Some words on words}
An {\sl alphabet} is a nonempty set $\Sigma$.
A finite vector $w=(w_1,\dots,w_n)\in\Sigma^n$ is called a {\sl word} of {\sl length} $n$ 
over $\Sigma$. We write $|w|=n$. Sequences $(w_1,w_2,\dots)$ and $(\dots,w_{-2},w_{-1})$ are 
{\sl one-sided infinite words} over $\Sigma$ and $(\dots,w_{-2},w_{-1},w_0,w_1,w_2,\dots)$ 
is a {\sl two-sided infinite word} over $\Sigma$ when $w_i\in\Sigma$ for all $i$.
The word of length zero is denoted by $\eps$ and is called {\sl the empty word}.

Let $\Sigma^*:=\cup_{n=0}^\infty\Sigma^n$ denote the set of all finite words over $\Sigma$.
Moreover, for an infinite index set $\I\in\{\Z,\N,-\N,\Z_+,\Z_-\}$, let $\Sigma^\I$ denote the set
of all infinite words $(w_n)_{n\in\I}$ over $\Sigma$.

The word $(w_1,w_2,\dots,w_n)$ is often simply written as $w_1w_2\dots w_n$.
Similarly for infinite words. For two words $u=u_1\dots u_m$ and $v=v_1\dots v_n$, 
the word $u_1\dots u_mv_1\dots v_n$ is denoted by $u\circ v$ or just $uv$.
This operation, called concatenation, is associative on $\Sigma^*$, with $\eps$
as the neutral element of $\Sigma^*$. Concatenation is also defined between two
oppositely directed one-sided infinite words (at their finite endpoints) and between
finite and one-sided infinite words in the natural way.

A word $w$ is called a {\sl subword} (or {\sl factor}) of a word $u$ if $u$
can be written as $xwy$ with (possibly empty) words $x$ and $y$. 
We write $w\prec u$ if $w$ is a subword of $u$. $\eps\prec u$ holds for all words $u$.
The reversed word of $u=u_1\dots u_m$ and $w=w_1w_2\dots$ is
$u^R:=u_m\dots u_1$ and $w^R:=\dots w_2w_1$, respectively. 

\subsection{Finite Fibonacci words: substitution, recursion and limit} \label{sec:finiteFib}
%One approch to the positive half $v_+:=(v_n)_{n\in\N}$ of the Fibonacci word $v=(v_n)_{n\in\Z}$ 
%from \eqref{eq:mod1} is via finite words:
%
Let $\Sigma=\{0,1\}$ and $\ph:\Sigma^*\to\Sigma^*$ be the homomorphism (w.r.t.~concatenation $\circ$)
with $\ph:0\mapsto 1$ and $\ph:1\mapsto 10$. Then put $f_1:=1$, $f_2:=\ph(f_1)$, $f_3:=\ph(f_2)$, etc.
In particular, we get
\[
\begin{array}{rclcl}
f_1&=&\ph(0)&=&1,\\
f_2&=&\ph(1)&=&10,\\
f_3&=&\ph(10)=\ph(1)\ph(0)&=&101,\\
f_4&=&\ph(101)=\ph(1)\ph(0)\ph(1)&=&10110,\\
f_5&=&\ph(10110)=\ph(1)\ph(0)\ph(1)\ph(1)\ph(0)&=&10110101,\\
f_6&=&\ph(10110101)=\ph(1)\ph(0)\ph(1)\ph(1)\ph(0)\ph(1)\ph(0)\ph(1)&=&1011010110110,\\
&\vdots
\end{array}
\]
This leads to the list of {\sl finite Fibonacci words} $f_1,f_2,\cdots$.
It is easy to see (by induction) that
\begin{equation}\label{eq:FiboRule}
 f_{n+1}=f_nf_{n-1}
\end{equation}
holds for $n\ge 2$, so that the length of $f_n$ is the $n$-th Fibonacci number; let us denote it by $F_n$.

The pointwise limit of this sequence $(f_n)$ is the one-sided infinite Fibonacci word 
$v_+=(v_n)_{n\in\N}$ with each $v_n$ from \eqref{eq:mod1}.
More precisely, equip $\Sigma$ with the discrete topology, $\Sigma^\N$ with the product
topology and extend each $f_n$ (by anything) to the right to a word in $\Sigma^\N$; 
then $(f_n)$ converges, by \eqref{eq:FiboRule}, and the limit is 
$v_+=1011010110110101101011011010110110\cdots$.

\subsection{The rotation formula and symmetry}
The above mechanisms define the positive half $v_+=v_1v_2\cdots$ of the two-sided infinite
Fibonacci word $v=(v_n)_{n\in\Z}$. The missing entries $\cdots,v_{-2},v_{-1},v_0$ can,
of course, be computed from the ``rotation formula'' \eqref{eq:mod1} but
they can also be expressed in terms of $v_+$:

For all $n\in\Z$, put 
\[
t_n:=n\alpha\!\!\!\mod 1\in[0,1),\quad\text{where}\quad \alpha=\tfrac{\sqrt 5-1}2
\quad\text{is the golden ratio,}
\]
so that $v_n = \chi_{[1-\alpha,1)}(t_n)$ by \eqref{eq:mod1}.
For arithmetics modulo $1$ it is of course useful to think of the interval $[0,1)$ as a circle with $0\cong1$.

Because $t_{-1}=1-\alpha$ and $t_0=0\cong1$ exactly mark the two endpoints of the interval
$[1-\alpha,1)$ and the sequences $(t_n)_{n\le-1}$ and $(t_n)_{n\ge0}$ evolve from
there, equispaced in opposite directions along our circle, one observes the symmetry $v_{-2}=v_1$, $v_{-3}=v_2,\ \dots$ in short:
\begin{equation} \label{eq:vsym}
v=v_+^R10v_+,
\end{equation}
where the $10$ in the middle refers to $v_{-1}v_0$.

Note that, by the irrationality of $\alpha$, all $t_n$ are pairwise distinct.
So the asymmetry that is caused by the different brackets of the interval $[1-\alpha,1)$
only shows for $n=-1$ and $n=0$, where $t_n$ exactly hits the two interval endpoints.
For $n\in\Z\setminus\{-1,0\}$, one has $v_n=v_{-1-n}$.

\subsection{Subwords of length $n$}
Another intrigiuing feature of the Fibonacci word is its small number of subwords.

Let $\Sigma=\{0,1\}$.
A random word $u\in\Sigma^\Z$ would, almost surely, contain every
one of the $2^n$ words $w\in\Sigma^n$ as a subword, for every $n\in\N$.
For the Fibonacci word $v\in\Sigma^\Z$, the situation is very different:
\begin{equation} \label{eq:subwords}
\begin{array}{|c|l|c|}
\hline
\text{length}&\text{subwords of $v$ of that length}&\text{count}\\\hline
1&0,\ 1&2\\
2&01,\ 10,\ 11&3\\
3&010,\ 011,\ 101,\ 110&4\\
4&0101,\ 0110,\ 1010,\ 1011,\ 1101&5\\
\vdots&&\vdots\\
n&\cdots&n+1\\\hline
\end{array}
\end{equation}
The number, say $sub_v(n)$, of subwords of $v$ of any length $n\in\N$ is exactly $n+1$.

For general words $u\in\Sigma^\Z$, it is easy to see that $sub_u$ grows monotonically, and
if $sub_u(n)=sub_u(n+1)$ for some $n$ then $sub_u(m)$ will remain at that value,
say $p$, for all $m\ge n$. The latter says that $v$ is $p$-periodic (up to a finite perturbation). 

So for an aperiodic word $u$, the function $sub_u$ grows strictly monotonically (by at least $1$
for each $n$), starting from $sub_u(1)=|\Sigma|=2$. So the subword count function with 
minimal growth (among the unbounded functions) is given by $sub_u(n)=n+1$. 
This is exactly what is observed for the Fibonacci word $u=v$.

%%%%%%%%%%%%%%%%%%%%%%%%%%%%%%%%%%%%%%%%%%%%%%%%
\section{Finite sections of the Fibonacci Hamiltonian}
Let $v=(v_n)_{n\in\Z}$ be the Fibonacci sequence \eqref{eq:mod1} and let
\begin{equation} \label{eq:FiboH}
 A := S_{-1}+M_{v}+S_1\ :\ \ell^p(\Z)\to\ell^p(\Z)
\end{equation}
be the Fibonacci Hamiltonian \eqref{eq:discrSchr}, where
\[
 S_k:\ell^p(\Z)\to\ell^p(\Z),\qquad (S_kx)_{n+k}=x_n,\qquad n\in\Z
\]
denotes the shift by $k\in\Z$ components and
\[
 M_b:\ell^p(\Z)\to\ell^p(\Z),\qquad (M_bx)_n=b_nx_n,\qquad n\in\N
\]
denotes the operator of pointwise multiplication by $b=(b_n)_{n\in\Z}\in\ell^\infty(\Z)$.

We identify $A$ with its two-sided infinite matrix $(a_{ij})_{i,j\in\Z}$ with
$a_{nn}=v_n$ and $a_{n,n\pm1}=1$ for all $n\in\Z$ and zeros everywhere else.
Corollary \ref{cor:fullFSM} connects the FSM of $A$ with the
limit operators of $A$. So we need to get a hand on these limit operators.

\subsection{Limit operators of the Fibonacci Hamiltonian}
Let $h=(h_1,h_2,\dots)$ be a sequence in $\Z$ with $h_k\to\pm\infty$, so that the limit operator $A_h$
of the Fibonacci Hamiltonian $A$ from \eqref{eq:FiboH} exists. Then
\[
A_h\ =\ (S_{-1})_h+(M_v)_h+(S_1)_h\ =:\ S_{-1}+M_{v_h}+S_1
\]
with a new potential 
\[
v_h\ :=\ \lim_{k\to\infty} S_{-h_k}v,
\]
where the limit is taken w.r.t.~pointwise convergence on $\Sigma^\Z$ for $\Sigma=\{0,1\}$.

%\Marko{Ich hab \eqref{eq:limops_expl} leider doch in der Literatur gefunden. Aber zumindest lag ich richtig. Vielleicht kann ich ``meinen'' Beweis ja noch ``for the reader's convenience'' einfuegen.}
%
The set $\cF$ of all such potentials $v_h$ is translation invariant (translations of limit operators of $A$ 
are limit operators of translations of $A$) and closed under pointwise convergence.
$\cF$ is the so-called {\sl Fibonacci subshift}. By our definition of $\cF$,
\[
\opsp(A)=\{S_{-1}+M_{v_h}+S_1\ :\ v_h\in\cF\}.
\]
The set $\cF$ is explicitly known (see e.g. Theorem 2.14 in \cite{Damanik2007}
and the appendix of \cite{DamanikLenz2004}):

\begin{equation} \label{eq:limops_expl}
\cF=\{v^\theta,w^\theta:\theta\in[0,1)\},
\end{equation}
where
\[
v_n^{\theta}\ :=\ \chi_{[1-\alpha,1)}(\theta + n\alpha\!\!\!\!\mod 1),\qquad
w_n^{\theta}\ :=\ \chi_{(1-\alpha,1]}(\theta + n\alpha\!\!\!\!\mod 1),\qquad n\in\Z.
\]
In particular, 
\begin{equation} \label{eq:selfsim}
A\in\opsp(A),
\end{equation}
since $v=v^0\in\cF$.
In fact, we do not need this explicit description \eqref{eq:limops_expl} of $\cF$.
The following lemma is sufficient (and much more handy) for us.
It expresses the well-known minimality of the Fibonacci subshift.
\begin{lemma}\label{lem:samewords}
Every $v_h\in\cF$ has the same list \eqref{eq:subwords} of subwords as $v$.
So for every $w\in\Sigma^*$ and every $v_h\in\cF$ it holds that
\[
w\prec v\quad\iff\quad w\prec v_h.
\]
\end{lemma}
\begin{proof}
Take arbitrary $w\in\Sigma^*$ and $v_h\in\cF$. So there is a sequence
$h=(h_1,h_2,\dots)$ in $\Z$ with $h_k\to\pm\infty$
and $v_h=\lim_{k\to\infty} S_{-h_k}v$, pointwise.

\fbox{$\Leftarrow$} If $w\prec v_h$ then $w\prec S_{-h_k}v$ for large $k$
(strict topology on $\Sigma$), so that $w\prec v$.

\fbox{$\Rightarrow$} Now let $w\prec v$. 
W.l.o.g.~assume $w\prec v_+$. Choose $n\in\N$ so that $w$ appears 
in the first $F_n$ letters of $v_+$, i.e. $w\prec f_n\prec f_{n+1}$
(recall the notations from \S\ref{sec:finiteFib}).

By \eqref{eq:FiboRule}, we have $f_{n+2}=f_{n+1}f_n$ and
$f_{n+3}=f_{n+2}f_{n+1}=f_{n+1}f_nf_{n+1}$. By induction, every $f_m$ with $m\ge n$, 
and hence $v_+$, is composed of $f_n$
and $f_{n+1}$. Since $w$ appears as a subword in $f_n$ and $f_{n+1}$, it appears
infinitely often in $v_+$, where two appearances of $w$ are at most $|f_{n+1}|=F_{n+1}$
letters away from each other. So every translate $S_{-h_k}v$ of $v$ contains $w$
in an $F_{n+1}$-neighbourhood of zero. Hence, every limit $v_h\in\cF$ contains $w$ 
(in an $F_{n+1}$-neighbourhood of zero).
\end{proof}

\subsection{Main results}
Now we are ready to state and prove our two main results.
\begin{theorem} \label{thm:A}
The FSM of the two-sided infinite Fibonacci Hamiltonian \eqref{eq:FiboH} is stable for any choice of cut-off points
and in every space $\ell^p(\Z)$ with $p\in[1,\infty]$.
\end{theorem}
The compression $A_\N$ of $A$ from \eqref{eq:FiboH} to $\ell^p(\N)$ is called {\sl one-sided infinite Fibonacci Hamiltonian}.
Its matrix $(a_{ij})_{i,j\in\N}$ is the submatrix of $A$ consisting of all rows and columns with $i,j\in\N$.

\begin{theorem} \label{thm:A+}
The FSM of the one-sided infinite Fibonacci Hamiltonian $A_\N$ is stable for any choice of cut-off points
and in every space $\ell^p(\N)$ with $p\in[1,\infty]$.
\end{theorem}
The rest of this paper is devoted to the proof of these two theorems.
The main ingredient, besides Corollary \ref{cor:fullFSM} and Lemma \ref{lem:Fredholm}, is the following lemma.
\begin{lemma} \label{lem:inj}
For the Fibonacci Hamiltonian $A$ from \eqref{eq:FiboH}, the following statements hold:\\[-2em]
\begin{enumerate}[\hspace{5mm}\bf a)]\itemsep-1mm
\item All $B\in\opsp(A)$ are injective on $\ell^\infty(\Z)$.
\item For all $B\in\opsp_+(A)$, the compression $B_-$ is injective on $\ell^\infty(\Z_-)$.
\item For all $B\in\opsp_-(A)$, the compression $B_+$ is injective on $\ell^\infty(\Z_+)$.
\end{enumerate}
\end{lemma}
Here we use the notations $B_\pm$ and $\Z_\pm$ from Definition \ref{def:ab} b).
We now show how this lemma implies Theorems \ref{thm:A} and \ref{thm:A+}
before we come to its proof (in Section \ref{sec:proofinj}).

\begin{proofof}{Theorem \ref{thm:A}}
Let $p\in[1,\infty]$. By Corollary \ref{cor:fullFSM}, we have to show that\\[-2em]
\begin{enumerate}[\hspace{5mm}\bf 1)]\itemsep-1mm
\item $A$ is invertible on $\ell^p(\Z)$,
\item for all $B\in\opsp_+(A)$, the compression $B_-$ is invertible on $\ell^p(\Z_-)$, and
\item for all $B\in\opsp_-(A)$, the compression $B_+$ is invertible on $\ell^p(\Z_+)$.
\end{enumerate}
It is sufficient to study the case $p=2$ as $A$ and all $B_+$ and $B_-$ are band matrices,
and so their invertibility is independent of $p\in[1,\infty]$ (see e.g. \cite[\S5.2.7]{Kurbatov}).

Property {\bf a)} of Lemma \ref{lem:inj} implies the invertibility of all $B\in\opsp(A)$,
by Lemma \ref{lem:Fredholm}. Since $A\in\opsp(A)$, by \eqref{eq:selfsim}, also $B=A$ is invertible.
So {\bf 1)} is shown.

To show {\bf 2)}, take an arbitrary $B\in\opsp_+(A)$ and look at $B_-$ as an operator on $\ell^2(\Z_-)$.
Since $B_-$ is injective on $\ell^\infty(\Z_-)$, by Lemma \ref{lem:inj} {\bf b)}, it is also
injective on the subset $\ell^2(\Z_-)$ of $\ell^\infty(\Z_-)$. 
Its adjoint is also injective on $\ell^2(\Z_-)$ since
$B_-$ is self-adjoint (by $A=A^*$). So it remains to show that the range of $B_-$ is closed:
From {\bf 1)} it follows that $A$ is Fredholm. By Lemma \ref{lem:Fredholm}, $B$ is invertible,
hence Fredholm. By Lemma \ref{lem:Fredholm} again, all operators in
$\opsp(B)\supset\opsp(B_-)$
are invertible, whence also $B_-$ is Fredholm (by Lemma \ref{lem:Fredholm} again)
and hence has a closed range.

{\bf 3)} follows from Lemma \ref{lem:inj} {\bf c)} and Lemma \ref{lem:Fredholm} in the very same way.
\end{proofof}

\begin{proofof}{Theorem \ref{thm:A+}}
This time we have to show that\\[-2em]
\begin{enumerate}\itemsep-1mm
\item[\hspace{5mm}\bf 4)] $A_\N$ is invertible on $\ell^p(\Z)$,
\item[\hspace{5mm}\bf 5)] for all $B\in\opsp_+(A_\N)$, the compression $B_-$ is invertible on $\ell^p(\Z_-)$.
\end{enumerate}
Statement {\bf 4)} follows from {\bf 3)} because $A_\N=B_+$ for $B=S_{-1}AS_1\in\opsp_-(A)$.\\
Statement {\bf 5)} follows from {\bf 2)} because $\opsp_+(A_\N)=\opsp_+(A)$, by the construction of $A_\N$.
\end{proofof}

\bigskip
Let us point out that the presence of (iii) in Lemma \ref{lem:Fredholm} is vital here.
With only (i) and (ii) at hand, we would be stuck in a vicious circle. The study of the invertibility of $A$ 
can be reduced to the following, presumably easier, problems: injectivity of $A$,
injectivity of $A^*$, Fredholmness of $A$. The latter again splits into many, presumably
easier, problems: invertibility of all limit operators $B$ of $A$, by Lemma \ref{lem:Fredholm} (ii). 
But now $A$ is one of those operators $B$, by \eqref{eq:selfsim}, which brings us back to the original problem!
So it is good to have -- and use -- Lemma \ref{lem:Fredholm} (iii) instead of (ii) here.

\medskip
Now all that remains to be done is the proof of Lemma \ref{lem:inj}.

\subsection{Proof of Lemma \ref{lem:inj}} \label{sec:proofinj}
First notice that one can restrict consideration to real sequences in both the one- and two-sided infinite case. 
Since $B$ (and the compressions $B_+$ and $B_-$) correspond to real matrices, it holds% for the kernel of $B$ 
\[
Bx=0\ \iff\ 0=\Re(Bx)= B(\Re(x)) \ \text{ and }\ 0=\Im(Bx)=B(\Im(x))
\]
with $\Re(\cdot)$ and $\Im(\cdot)$ denoting the real and imaginary part of a sequence. 
So the injectivity of $B$ on the space of real bounded sequences implies the injectivity on the 
space $\ell^\infty(\I)$ of complex bounded sequences. 
One is left with proving $Bx=0 \Rightarrow x=0$ for all bounded real sequences.
The idea is most transparent in the one-sided infinite case.
So let us start with the proof of {\bf c)}.

To show that an operator $B_+$ is injective on $\ell^\infty(\Z_+)$, derive
the entries $x_1, x_2,\dots$ of a solution $x=(x_n)_{n\in\Z_+}$ of the 
homogeneous equation $B_+x=0$, starting from a nonzero initial entry $x_0$, 
and prove that some entry $x_n$ will eventually
exceed (in modulus) any previously given bound $r>0$. Because, for every $r>0$,
this computation will only take finitely many steps $x_1,\dots,x_n$, 
it is enough to know about finite subwords of the potential of $B_+$. 
(Our proof does not use the explicit formula \eqref{eq:limops_expl}.)

%\Marko{Achtung: $\Z_+$ faengt mit $0$ an.}
Identify $B_+$ with its matrix $(b_{ij})_{i,j\in\Z_{+}}$. Because of the tridiagonal structure,
the value of $x_0$ is sufficient to calculate the whole solution vector $x$. More precisely, 
$x_1=-b_{00}\,x_0$ and $x_{n+1}=-b_{nn}\,x_n-x_{n-1}$ for $n \in \N$. 
As usual, rewrite this recurrence with transfer matrices:
\[
T_{b_{nn}}
\begin{pmatrix}
x_{n-1} \\ x_n
\end{pmatrix}
=
\begin{pmatrix}
x_n \\ x_{n+1}
\end{pmatrix},
\quad\text{where}\quad
T_{b_{nn}}=\begin{pmatrix}
0 & 1 \\
-1& -b_{nn}
\end{pmatrix}
\quad\text{with}\quad
b_{nn}\in\Sigma=\{0,1\}.
\]
%With the diagonal elements $b_{nn}\in\Sigma=\{0,1\}$ for all $n\in\Z_+$ the sequence $x$ is determined by the two matrices $T_0$ for $b_{nn}=0$ and $T_1$ for $b_{nn}=1$. Let $b_n=b_{nn}$ be the diagonal of $b$. 
W.l.o.g.~we can assume $x_0=1$.
Here is an example computation for a certain diagonal $(b_{nn})$: %and $n=0,1,\dots,16$:
\\[-0.7em]
\definecolor{hellgrau}{rgb}{0.85,0.85,0.85}
\newcommand{\hgr}[1]{\colorbox{hellgrau}{$#1$}}
\[
\tabcolsep-1mm
\begin{array}{|r|cccccccccccccccccccc}
\hline
n&0&1&2&3&4&5&6&7&8&9&10&11&12&13&14&15&16&\cdots\\
\hline
b_{nn}&1&0&\multicolumn{1}{c|}{1}&1&0&\multicolumn{1}{c|}{1}&0&\multicolumn{1}{c|}{1}&1&0&\multicolumn{1}{c|}{1}&1&0&\multicolumn{1}{c|}{1}&0&\multicolumn{1}{c|}{1}&1&\cdots\\
\hline
x_n&\hgr1&-1&-1&\hgr2&-1&-2&\hgr3&2&\hgr{-5}&3&5&\hgr{-8}&3&8&\hgr{-11}&-8&\hgr{19}&\cdots\\
\hline
\end{array}
\]
In this example, with a bit of optimism, we seem to observe that\\[-2em]
\begin{itemize}\itemsep-1mm
\item the diagonal $(b_{nn})$ of $B$ is composed of blocks ``101" and ``01", and
\item the entries $\hgr{x_n}$ at the beginning of each block grow unboundedly in modulus.
\end{itemize}
We will prove that this is always the case.
The following lemma is a special case of a partition of general Sturmian words
as introduced in \cite{LenzPhD} (also see \cite{DamanikLenz1999,Lenz03}):
\begin{lemma}\label{lem:101,01}
The diagonal $b:=(b_{nn})_{n\in\Z_+}$ of $B_+$ with $B\in\opsp_-(A)$ is of the form
\[
b=p\,w_1\,w_2\,w_3\cdots
\qquad\text{with}\qquad
p\in\{\eps,1\}
\quad\text{and}\quad
w_i\in\{101,01\}
\quad\text{for all}\quad
i\in\N.
\]
\end{lemma}
\begin{proof}
By Lemma \ref{lem:samewords} and \eqref{eq:subwords}, $b$ contains neither 00 nor 111 as a subword.
So $0$ is always followed by $1$, and $1$ is always followed by $101$ or $01$.
\end{proof}

So we are particularly interested in the patterns ``101" and ``01" and their corresponding transfer matrices
\[
T_{101}:=T_1T_0T_1=\begin{pmatrix}0 & -1 \\ 1 & 2\end{pmatrix}
\qquad\text{and}\qquad
T_{01}:=T_1T_0=\begin{pmatrix}-1 & 0 \\ 1 & -1\end{pmatrix}.
\]

%Recalling the factors of the Fibonacci word $v$ and using Lemma 4.1 we can write $b=pw_1w_2w_3 \ldots$ with a prefix $p \in \{\varepsilon, 1\}$ and $w_i \in \{101, 01\}$. Using that neither $00$ nor $111$ is a subword of b, a double $1$ is always followed by a $0$, as well as a single $0$ is always followed by a $1$. The prefix $\varepsilon$ covers the case, in which $b$ start with a "0" or a single "1" whereas $p=1$ covers the opurtunity of $b$ starting with $11$. After each block $w_i$ we are in a position of a single $1$, where only a "01" or a "101" can follow. These blocks strongly correspond to the growth of the solution x of the homogenous system. 

Let us say that a vector $\vekz{y_1}{y_2}\in\R^2$ has property {\bf C} if
$y_1\cdot y_2<0$ and $|y_1|<|y_2|$.

\begin{lemma} \label{lem:preserveC}
Both $T_{101}$ and $T_{01}$ preserve property {\bf C}.
More precisely, if $\vekz{y_1}{y_2}\in\R^2$ has property {\bf C} then
$\vekz{z_1}{z_2}:=T_w\,\vekz{y_1}{y_2}$ with $w\in\{101,01\}$ has properties
%\begin{itemize}
\begin{enumerate}[\hspace{20mm}\bf A)]\itemsep-1mm
\item $|z_2|>|y_2|$ with $|z_2|-|y_2|\ge\min\{|y_1+y_2|, |y_1|\}>0$,
\item $|z_1+z_2|\ge|y_1+y_2|$ and $|z_1|\ge|y_1|$,\quad and
\item $z_1\cdot z_2<0$ and $|z_1|<|z_2|$.
\end{enumerate}
%\end{itemize}
\end{lemma}
% Supposing a growth condition $a_1\cdot a_2 < 0 \text{ and } |a_1| <|a_2|$ for the equation system 
%\[
% T_i \begin{pmatrix}a_1 \\ a_2\end{pmatrix} 
% =
% \begin{pmatrix}b_1 \\ b_2\end{pmatrix}
%\] with $i\in \{101, 01\}$ it holds  
%\begin{enumerate}[\hspace{5mm}\bf 1)]\itemsep-1mm
%\item $|b_2|>|a_2|$
%\item $|b_1 + b_2 |\geq |a_1 + a_2| \ \ \text{and}  \ \ |b_1| \geq |a_1|$
%\item $b_1\cdot b_2 <0 \text{ and } |b_1| < |b_2|$
%\end{enumerate}
%\[
%a_1\cdot a_2 < 0 \text{ and } |a_1| <|a_2| \Rightarrow |b_2|>|a_2| .
%\]
\begin{proof}
This is a straightforward computation using%follows directly from evaluating the equations
\[T_{101}
\begin{pmatrix}
y_1\\y_2
\end{pmatrix}
=
\begin{pmatrix}
-y_2 \\y_2 + (y_1 + y_2)
\end{pmatrix}
\qquad\text{and}\qquad
 T_{01} 
\begin{pmatrix}
y_1 \\ y_2
\end{pmatrix}
= 
\begin{pmatrix}
-y_1 \\y_1 - y_2
\end{pmatrix}.\]
\end{proof}

Property {\bf A} shows a growth (in modulus) of the second vector component 
after applying $T_{101}$ or $T_{01}$.
% is applied to, in this case the last known entry of $x$. 
By property {\bf B}, %says that, if multiple transfer matrices are applied, 
the amount of growth is %uniformly bounded from below by $\min\{|y_1+y_2|, |y_1|\}>0$.
non-decreasing when applying $T_{101}$ or $T_{01}$ again.
%, while property \textbf{3)} tells us, that the growth condition is transferred to the vector 
%$\begin{pmatrix}b_1 & b_2 \end{pmatrix}^T$.
%\[
%a_1\cdot a_2 < 0 \text{ and } |a_1| \leq|a_2| \Rightarrow |b_1 + b_2 |\geq |a_1 + a_2| \ \ \text{and}  \ \ |b_1| \geq |a_1|.
%\]
%Furthermore we have
%\[
% a_1\cdot a_2 < 0 \text{ and } |a_1| <|a_2| \Rightarrow b_1\cdot b_2 <0 \text{ and } |b_1| < |b_2|.
%\]
The fact that property {\bf C} is preserved keeps the argument working for the next application
of $T_{101}$ or $T_{01}$, leading to unbounded growth of the second vector component, by induction.

So all that we need is one first ocurrence of property {\bf C} for a vector $\vekz{x_n}{x_{n+1}}$
in our computation of a sequence $x=(x_0,x_1,x_2,\dots)$ that solves $B_+x=0$.

%Therefore by supposing $a_1\cdot a_2 < 0 \text{ and } |a_1| <|a_2|$ the sequence \[
%h_n = \Big|\Big(\prod_{j=1}^{n}T_j 
%\begin{pmatrix}
%a_1 \\ a_2
%\end{pmatrix}\Big)_2 \Big| 
%\ \ \ T_j \in \{ T_{101},T_{01} \}
%\] is strictly monotonic with $|h_n-h_m|\geq (n-m)\cdot \min\{a_1+a_2, a_1\}$ and therefore unbounded.
%Now apply these thoughts on $B_+x=0$. W.l.o.g look at $x_0=1$ and compute the first digits of the solution $x$ manually. 

We start with the case $p=\varepsilon$. Besides $x_0=1$ (see above), we put $x_{-1}:=0$ to
start our recurrence and account for the non-existence of column number $-1$ in the matrix $B_+$.
Depending on which of $T_{101}$ and $T_{01}$ we apply to $\vekz{x_{-1}}{x_0}=\vekz01$, we get
%is the first entry of $x$, one can use the transfer matrices by looking at the vector $\begin{pmatrix}0 & x_0\end{pmatrix}^T=\begin{pmatrix}
%0 & 1 \end{pmatrix}^T$ which takes the nonexistence of $x_{-1}$ into account.
%For $p=101$ the first digits of x compute to $x = \begin{pmatrix} 1 & -1 & -1 & 2 & \ldots \end{pmatrix}$. The sequence
%\[ 
% h_n = \Big|\Big(\prod_{j=1}^{n}T_{w_j }
% \begin{pmatrix}
% x_3 \\ x_4
% \end{pmatrix}\Big)_2
% \Big|
%\]
%is, using that the growth condition is satisfied, unbounded. Therefore $x$ has an unbounded subsequence and is thus unbounded itself.
%For $p=01$ the first digits of x compute to $\begin{pmatrix}1 & 0 & -1  & \ldots \end{pmatrix}$.
%The application of the transfer matrices yields 
\[
 T_{01}\begin{pmatrix} 0 \\ 1 \end{pmatrix} = \begin{pmatrix} -1&	0\\ 1&	-1  \end{pmatrix}  \begin{pmatrix} 0 \\ 1 \end{pmatrix} = \begin{pmatrix} 0 \\ -1 \end{pmatrix}
 \qquad\text{or}\qquad
 T_{101}\begin{pmatrix} 0 \\ 1 \end{pmatrix} = \begin{pmatrix} 0&	-1\\ 1&	2  \end{pmatrix}  \begin{pmatrix} 0 \\ 1 \end{pmatrix} = \begin{pmatrix} -1 \\ 2 \end{pmatrix}.
\] 
So repeated application of $T_{01}$ leads to $\pm\vekz01$ but after the first application of $T_{101}$,
which will eventually happen since $b$ is not periodic, one gets $\vekz{x_n}{x_{n+1}}=\pm\vekz{-1}2$
for some $n\in\N$. This vector has property {\bf C}. From our arguments above it follows that the sequence $x$ is unbounded.

% and since b is not periodic, there exists a smallest index $i_0$ with $w_{i_0}=101$ so that
%\[
% \prod_{j=1}^{i_0}T_{w_j} 
% \begin{pmatrix}
% 0 \\ x_0
% \end{pmatrix}
% =\pm\begin{pmatrix}
% -1 \\ 2
% \end{pmatrix}.
%\]
%Hence the sequence
%\[ 
% h_n = \Big|\Big(\prod_{j=1}^{n}T_{w_j} 
% \begin{pmatrix}
% 0 \\ x_0
% \end{pmatrix}\Big)_2
% \Big|
%\]
%is strictly monotonic for $n>i_0$ and unbounded. It follows, that $x$ is unbounded.
If the prefix $p$ of $b$ is $1$, it follows from $x_0=1$ that $x_1=-1$, so that our
recurrence starts with $\vekz{x_0}{x_1}=\vekz1{-1}$.
%the first relevant digits of x compute to $\begin{pmatrix}1 & -1 & \ldots \end{pmatrix}$. 
The application of the transfer matrices yields 
\[
 T_{01}\begin{pmatrix} 1 \\ -1 \end{pmatrix} = \begin{pmatrix} -1&	0\\ 1&	-1  \end{pmatrix}  \begin{pmatrix} 1 \\ -1 \end{pmatrix} = \begin{pmatrix} -1 \\ 2 \end{pmatrix}
 \qquad\text{or}\qquad
 T_{101}\begin{pmatrix} 1 \\ -1 \end{pmatrix} = \begin{pmatrix} 0&	-1\\ 1&	2  \end{pmatrix}  \begin{pmatrix} 1 \\ -1 \end{pmatrix} = \begin{pmatrix} 1 \\ -1 \end{pmatrix}.
\]  
So repeated application of $T_{101}$ leads to $\vekz1{-1}$ but after the first application of $T_{01}$,
which will eventually happen since $b$ is not periodic, one gets $\vekz{x_n}{x_{n+1}}=\vekz{-1}2$
for some $n\in\N$. This vector has property {\bf C}. From our arguments above it follows that the sequence $x$ is unbounded.

%and since again b is not periodic, there exists a smallest index $j_0$ with $w_{j_0}=01$ so that it holds
%\[ \prod_{j=1}^{j_0}T_{w_j} 
% \begin{pmatrix}
% x_1 \\ x_2
% \end{pmatrix}
% =\begin{pmatrix}
% -1 \\ 2
% \end{pmatrix}.
%\]
%Hence the sequence
%\[
% h_n=\Big|\Big(\prod_{j=1}^{n}T_{w_j} 
% \begin{pmatrix}
% x_1 \\ x_2
% \end{pmatrix}\Big)_2
% \Big| 
%\]
%is strictly monotonic for $n>j_0$ thus $x$ is unbounded.

For both possibilities of the prefix $b\in\{\eps,1\}$ and all possibilities of the following blocks $w_i\in\{101,01\}$
(recall Lemma \ref{lem:101,01}), we have shown that all nontrivial solutions $x$ of the homogenous system 
$B_+x=0$ are unbounded, so that every $B_+$ with $B\in\opsp_-(A)$ is injective on $\ell^{\infty}(\Z_+)$.

%...
%
%transfer matrix
%
%patterns 101 and 01 and their transfer matrices
%
%show that $v_h$ splits in blocks 101 and 01
%
%possible prefixes
%
%...

\medskip

To see {\bf b)}, consider the three flip operators
\[
J_\leftarrow:\ell^\infty(\Z_+)\to\ell^\infty(\Z_-),\qquad
J_\rightarrow:\ell^\infty(\Z_-)\to\ell^\infty(\Z_+),\qquad\text{and}\qquad
J_\leftrightarrow:\ell^\infty(\Z)\to\ell^\infty(\Z),
\]
all three acting by the rule $x\mapsto y$ with $y_n=x_{-n}$ for $n$ in $\Z_-$, $\Z_+$ and $\Z$, respectively.

The formula $v=v_+^R10v_+$ from \eqref{eq:vsym} implies that
$\opsp_+(A)$ exactly consists of the reflections $C=J_\leftrightarrow BJ_\leftrightarrow$
of operators $B\in\opsp_-(A)$, so that
\[
 \{C_-\ :\ C\in\opsp_+(A)\}\ =\ \{(\underbrace{J_\leftrightarrow BJ_\leftrightarrow}_C)_+=J_\leftarrow B_+J_\rightarrow\ :\ B\in\opsp_-(A)\}.
\]
So, clearly, {\bf b)} follows from {\bf c)}.

\medskip

Finally, for the proof of {\bf a)}, we use a two-sided version of
the proof of {\bf c)}.

Let $B\in\opsp(A)$ and again let $b\in\{0,1\}^\Z$ be the diagonal of $B$. Let $x$ be a nontrivial solution
of the homogeneous equation $Bx=0$.
%
%Using Lemma 4.1 the diagonal can be written as $b=\ldots w_{-4} w_{-3} {101}{\bf{101}} {01} {\bf{101}}{101}w_{3} w_{4}\ldots$ 
%
%\textcolor{blue}{\tt Argumentieren wir hier einfach, dass 10110101101101 ein 14-stelliges Teilwort von $v$ ist(?) und somit, laut Lemma 4.1, auch von $b$?}
%
%with already marked blocks and $w_i \in \{101, 10\}$ for $ i<0$ and $w_i \in \{101, 01\}$ for $i>0$.
%
%\textcolor{blue}{\tt Wie sieht man, dass die Bloecke nach links hin 101 und 10 sind? Argumentieren wir einfach aufs neue wie im Beweis von Lemma \ref{lem:101,01}? Oder benutzt man sowas wie \eqref{eq:vsym}?}
%
% Index the middle block by $x_0$ and $x_1$. 
% 
% 
%\textcolor{blue}{{\tt Also so hier?}
%\[
%b\ =\ \cdots\ w_{-4}\ w_{-3}\ {101}\ {\bf{101}} \underbrace{0}_{x_0}\underbrace{1}_{x_1} {\bf{101}}\ {101}\ w_{3}\ w_{4}\ \cdots
%\]
%{\tt Aber wahrscheinlich nicht $x_{0/1}$ sondern $b_{0/1}$?} 
%}
%
We will prove that the sequence $x$ grows unboundedly in at least one direction, left or right.
The growth to the right is studied as in the proof of {\bf c)} above -- growth to the left
by symmetric arguments. Here is the analogue of Lemma \ref{lem:101,01}.

\begin{lemma}\label{lem:101,01,10}
The diagonal $b:=(b_{nn})_{n\in\Z}$ of $B\in\opsp(A)$ is of the form
\begin{equation} \label{eq:101,01,10}
b\ =\ \cdots\ w_{-4}\ w_{-3}\ \underbrace{101}_{w_{-2}} \underbrace{101}_{w_{-1}} \underbrace{01}_{w_0} \underbrace{101}_{w_1}\underbrace{101}_{w_2}\ w_{3}\ w_{4}\ \cdots
\end{equation}
with $w_{-i}\in\{101,10\}$ and $w_i\in\{101,01\}$ for all $i\in\N$.
\end{lemma}
\begin{proof}
The word $101\,101\,01\,101\,101$ is contained in the Fibonacci word $v$ (as $v_{-6}\cdots v_7$)
and therefore, by Lemma \ref{lem:samewords}, also in $b$.
By Lemma \ref{lem:samewords} and \eqref{eq:subwords}, $b$ contains neither 00 nor 111 as a subword.
So, as argued in Lemma \ref{lem:101,01}, $0$ is always followed by $1$, and $1$ is always followed by $101$ or $01$. 
Moreover, $0$ is always preceded by $1$, and $1$ is always preceded by $101$ or $10$.
\end{proof}

So besides $T_{101}$ and $T_{01}$, we now also look at the transfer matrix $T_{10}:=T_0T_1$.
Note that, when we study the asymptotics of $x$ towards $-\infty$ (going backward in ``time''),
we will have to look at inverses of the transfer matrices.

Therefore, let us say that a vector $\vekz{y_1}{y_2}\in\R^2$ has property {\bf F} if
$y_1\cdot y_2<0$ and $|y_1|>|y_2|$. Here is the ``leftward'' analogue of Lemma \ref{lem:preserveC}.

\begin{lemma}
Both $T_{101}^{-1}$ and $T_{10}^{-1}$ preserve property {\bf F}.
More precisely, if $\vekz{y_1}{y_2}\in\R^2$ has property {\bf F} then
$\vekz{z_1}{z_2}:=T_w^{-1}\,\vekz{y_1}{y_2}$ with $w\in\{101,10\}$ has properties
%\begin{itemize}
\begin{enumerate}[\hspace{20mm}]\itemsep-1mm
\item[\bf D)] $|z_1|>|y_1|$ with $|z_1|-|y_1|\ge\min\{|y_1+y_2|, |y_2|\}>0$,
\item[\bf E)] $|z_1+z_2|\ge|y_1+y_2|$ and $|z_2|\ge|y_2|$,\quad and
\item[\bf F)] $z_1\cdot z_2<0$ and $|z_1|>|z_2|$.
\end{enumerate}
%\end{itemize}
\end{lemma}
\begin{proof}
This is again a straightforward computation using%follows directly from evaluating the equations
%Similar to above derivations and supposing a second growth condition $a_1\cdot a_2 < 0 \text{ and } |a_1| >|a_2|$ on the equation system
%\[
% T_i^{-1} \begin{pmatrix}a_1 \\ a_2\end{pmatrix} = \begin{pmatrix}b_1 \\ b_2\end{pmatrix}
%\]
%with $i\in\{101, 10\}$, it holds
%\begin{enumerate}[\hspace{5mm}\bf 1)]\itemsep-1mm
%\item $|b_1|>|a_1|$
%\item $|b_1 + b_2 |\geq |a_1 + a_2| \ \ \text{and}  \ \ |b_2| \geq |a_2|$
%\item $b_1\cdot b_2 <0 \text{ and } |b_1| > |b_2|$
%\end{enumerate}
%%
%%\[
%%a_1\cdot a_2 < 0 \text{ and } |a_1| >|a_2| \Rightarrow |b_1|>|a_1|
%%\]
%which again follows directly from evaluating the equations
\[
T_{101}^{-1}
\begin{pmatrix}
y_1 \\ y_2
\end{pmatrix}
=
\begin{pmatrix}
y_1+ y_2 + y_1 \\ -y_1
\end{pmatrix}
\qquad\text{and}\qquad
T_{10}^{-1} 
\begin{pmatrix}
y_1 \\ y_2
\end{pmatrix}
=
\begin{pmatrix}
-y_1 + y_2 \\ -y_2
\end{pmatrix}. \qedhere
\] 
\end{proof}

As before, property {\bf D} states a growth (in modulus) of the first component. Property {\bf E} ensures that
the amount of this growth is non-decreasing in further applications of $T_w^{-1}$ with $w\in\{101,10\}$, 
and the fact that property {\bf F} is preserved makes sure that the same argument keeps working for 
further applications of $T_w^{-1}$, leading to unbounded growth.

So again, we just need a first ocurrence of property {\bf F} for a vector $\vekz{x_n}{x_{n+1}}$
with $n<0$ or a first ocurrence of property {\bf C} for a vector $\vekz{x_n}{x_{n+1}}$ with $n\ge 0$
in our computation of a sequence $x=(\dots,x_{-2},x_{-1},x_0,x_1,x_2,\dots)$ that solves $Bx=0$.
Then $x$ will be unbounded.

%says that the growth condition is transferred to the vector $\begin{pmatrix} b_1 & b_2\end{pmatrix}$.
%
%Looking at the growth (in modulo) from $a_1$ to $b_1$ we observe under the conditions above, that stepwise growth is conserved
%\[
% a_1\cdot a_2 < 0 \text{ and } |a_1| >|a_2| \Rightarrow |b_1 + b_2 |\geq |a_1 + a_2| \ \ \text{und}  \ \ |b_2| \geq |a_2|.
%\]
%Furthermore we have
%\[
% a_1\cdot a_2 < 0 \text{ and } |a_1| >|a_2| \Rightarrow b_1\cdot b_2 <0 \text{ and } |b_1| > |b_2|.
%\]
%Thus, by supposing $a_1\cdot a_2 < 0$ and $|a_1| >|a_2|$, the sequence 
%\[
%h_n = |(\prod_{j=1}^{n}T_j^{-1} 
%\begin{pmatrix}
%a_1 \\ a_2
%\end{pmatrix})_{1}
%| \ \ \ T_j \in \{T_{101},T_{10}\}
%\]
%is strictly monotonic with $|h_n-h_m|\geq |n-m|\cdot\min\{a_1+a_2, a_2\}$ and therefore unbounded. 

This time, one entry, say $x_0$, does not determine the whole sequence $x$, but two entries do.
Let the two entries of $x$ that are associated to the entries $0$ and $1$ of $w_0$ in \eqref{eq:101,01,10} be equal
to $\alpha$ and $\beta$, respectively, with arbitrary $\alpha, \beta \in \R$. W.l.o.g~label them as $x_0$ and $x_1$.

Using the adjacent entries of $w_0=01$ in $b$, see \eqref{eq:101,01,10}, the corresponding entries in $x$ turn out to be
% one can compute the first entries of the solution directly with 
\[
\begin{pmatrix} x_{-4} & \cdots	& x_{5} \end{pmatrix}=
\begin{pmatrix} \alpha-2\beta 	& \beta 	& -\alpha+\beta & -\beta & \alpha 	& \beta & -\alpha-\beta & \alpha 	& \alpha+\beta 	& -2\alpha-\beta \end{pmatrix}.
\]

With respect to $\alpha$ and $\beta$, we have to distinguish the following cases:
\begin{enumerate}\itemsep0mm
\item If $\alpha = \beta=0$ then $x=0$ follows.
\item If $\alpha = 0$ and $\beta\neq 0$ then $\vekz{x_{-4}}{x_{-3}}=\vekz{-2\beta}\beta$ has property {\bf F}.
%and $x_{-4}=-2\beta$ the sequence 
% \[ h_n = 
% \Big|\Big(\prod_{j=2}^{n}T_{w_{-j}}^{-1} 
% \begin{pmatrix}
%	x_{-4} \\ x_{-3}
%	\end{pmatrix}\Big)_1
% \Big|
% \]
% is strictly monotonic and unbounded.
\item If $\alpha \neq 0$ and $\beta = 0$ then $\vekz{x_{4}}{x_{5}}=\vekz{\alpha}{-2\alpha}$ has property {\bf C}.
%With $x_{4}=\alpha$ and $x_{5}=-2\alpha$ the sequence 
%\[ h_n = \Big|\Big(\prod_{j=2}^{n}T_{w_j}
% \begin{pmatrix}
% x_{4} \\ x_{5}
% \end{pmatrix}\Big)_2
% \Big|
%\]
%is strictly monotonic and unbounded.
\item If $\alpha\neq 0$ and $\beta\neq 0$ we have to look at two more cases: 
	\begin{enumerate}
	\item If $\alpha \cdot\beta>0$ then $\vekz{x_{1}}{x_{2}}=\vekz{\beta}{-\alpha-\beta}$ has property {\bf C}.
%          With $x_1=\beta$ and $x_2=-(\alpha+\beta)$ the sequence
%	\[ h_n = \Big|\Big(\prod_{j=1}^{n}T_{w_j} 
%	 \begin{pmatrix}
%	 x_1 \\ x_2
%	 \end{pmatrix}\Big)_2
%	 \Big| 
%	\]
%	is strictly monotonic and unbounded.
	\item If $\alpha \cdot\beta<0$ then $\vekz{x_{-4}}{x_{-3}}=\vekz{\alpha-2\beta}\beta$ has property {\bf F}.
%          With $x_{-3}=\beta$ and $x_{-4}=\alpha-2\beta$ the sequence 
% 	\[ h_n = 
% 	 \Big|\Big(\prod_{j=2}^{n}T_{w_{-j}}^{-1} 
% 	 \begin{pmatrix}
%	 x_{-4} \\ x_{-3}
%	 \end{pmatrix}\Big)_1
% 	 \Big|
% 	\] 
% 	is strictly monotonic and unbounded.
	\end{enumerate}
\end{enumerate}
This completes the study of all cases.
Each nontrivial solution of the homogenous equation $Bx=0$ is unbounded, thus $B$ is injective on $\ell^\infty(\Z)$.
This completes the proof of Lemma \ref{lem:inj}.

%\Marko{Wie Du schon vorgeschlagen hast, sollten wir Lemma 4.4 gleich komplexwertig auffassen und im Beweis sagen, dass man sich (nach Zerlegung in Real- und Imaginaerteil) w.l.o.g. auf reellwertige Folgen zurueckziehen kann.}
%While focussing on purely real sequences, the application of the Fibonacci Hamiltonian lies in the quantum world, which builds on complex (sequence) spaces. Since the theorems from \S 2 work for complex vector spaces, one only needs to show a complex version of Lemma \ref{lem:inj}, where all the $\ell^p(\I)$ for $\I \in \{\Z, \Z_+, \Z_-\}$ are assumed to be complex sequence spaces.

%\begin{proofof}{Lemma \ref{lem:complinj}}
%This lemma follows mainly from the real case. For the kernel of $B$ holds $Bx=0 \Leftrightarrow Bx_{re} + i Bx_{im}=0$ with $x_{re}$ beeing the real part of $x$ and $x_{im}$ beeing the imaginary part. With $x_{re}$ and $x_{im}$ beeing real sequences Lemma \ref{lem:inj} implies directly $x=0$.
%\end{proofof}

%%%%%%%%%%%%%%%%%%%%%%%%%%%%%%%%%%%%%%%%%%%%%%%%
\medskip

{\bf Acknowledgements.}
The first author thanks the organizers, Albrecht B\"ottcher, Daniel Potts, Peter Stollmann and David Wenzel,
for a wonderful IWOTA conference 2017 in Chemnitz with many inspiring talks and countless other pleasant moments.

\end{document}